\documentclass{article}
\usepackage{a4wide}
\usepackage[utf8]{inputenc}
\usepackage{amsmath}
\usepackage{amsthm}
\usepackage{fullpage}

\newcommand{\APEP}{\mbox{\sc APEP}}
\newcommand{\PrBPP}{\PromiseBPP}

\usepackage{amsmath}
\usepackage{amsthm}
\usepackage{amssymb}
\usepackage{epsfig}
\usepackage{xcolor}

\usepackage{amsmath}
\usepackage{amsthm}
\usepackage{amssymb}
\usepackage{epsfig}
\usepackage{xcolor}
\newtheorem{theorem}{Theorem}[section]

\newtheorem{claim}{Claim}[theorem]
\newtheorem{proposition}{Proposition}
\newtheorem{observation}[theorem]{Observation}
\newtheorem{corollary}[theorem]{Corollary}

\theoremstyle{definition} 
\newtheorem{definition}[theorem]{Definition}

\newcommand{\algoapep}{\mbox{$A_{\rm ape}$}}

        \newcommand{\numberNP}{{\mathchoice{\raisebox{1pt}
        {$\displaystyle\#$}{\rm NP}}{\raisebox{1pt}
        {$\textstyle\#$}{\rm NP}}{\raisebox{1pt}
        {$\scriptstyle\#$}{\rm NP}}{\raisebox{1pt}
        {$\scriptscriptstyle\#$}{\rm NP}}}}

\newcommand{\sharpNP}{\numberNP}

\newcommand{\ignore}[1]{}

\newcommand{\newfontobj}[2]{
  \newcommand{#1}[1]{
    \expandafter\def\csname##1\endcsname{{#2 ##1}}}}

\newfontobj{\class}{\rm}
\newfontobj{\lang}{\bf}
\newfontobj{\oper}{\rm}

\newcommand{\poly}{\mbox{\tt poly}}

\class{PSPACE}          
\class{AM}              
\class{AMTIME}
\class{AMSUBEXP}
\class{MA}
\class{NP}
\class{AC}
\class{NLIN}
\class{QBF}
\class{P}
\class{RP}
\class{BPP}
\class{NPSV}
\class{advBPP}
\class{DTIME}
\class{ZPP}
\class{EXPSPACE}
\class{PH}
\class{IP}
\class{SNP}
\class{SNQP}
\class{AVP}
\class{QP}
\class{SNTIME}
\class{PrAM}
\class{PrMA}

\class{AvgTIME} \class{DistNP} \class{AZPP} \class{AVE}
\class{DistE} \class{DistEXP} \class{E} \class{NE} \class{UE}
\class{EXP} \class{pcomp} \class{NEXP} \class{UP} \class{DTIME}
\class{NTIME} \class{R} \class{NSUBEXP} \class{SUBEXP} \class{PF}
\class{LINcomp}
\class{BPTIME}
\class{DistNLIN}
\class{AME}
\class{SIZE}
\class{NSIZE}
\class{SAT}
\class{MCSP}
\class{SubExp}
\class{CP}
\class{SZK}
\class{PP}
\class{PromiseBPP}
\class{PromiseMA}
\class{CoRP}
\class{SearchBPP}
\class{BPTIME}

\title{Promise Problems Meet Pseudodeterminism
\thanks{Research supported in part by NSF grant 1934884}\\}
\author{Peter Dixon\thanks{Department of Computer Science, Iowa State University. tooplark@iastate.edu}, ~~A. Pavan\thanks{Department of Computer Science, Iowa State University. pavan@cs.iastate.edu}, ~~N. V. Vinodchandran\thanks{Department of Computer Science and Engineering, University of Nebraska, Lincoln. vinod@cse.unl.edu}}
\date{}

\begin{document}

\maketitle
\begin{abstract} The {\sc Acceptance Probability Estimation Problem} (\APEP) is to additively approximate the acceptance probability of a Boolean circuit. This problem admits a probabilistic approximation scheme. A central question is whether we can design a {\em pseudodeterministic} approximation algorithm for this problem: a probabilistic polynomial-time algorithm that outputs a canonical approximation with high probability. Recently, it was shown that such an algorithm would imply that {\em every approximation algorithm can be made pseudodeterministic} (Dixon, Pavan, Vinodchandran; {\em ITCS 2021}).    

The main conceptual contribution of this work is to establish that the existence of a pseudodeterministic algorithm for \APEP\ is fundamentally connected to the relationship between probabilistic promise classes and the corresponding standard complexity classes. In particular, we show the following equivalence: {\em every promise problem in\PromiseBPP\ has a solution in\BPP\ if and only if $\APEP$ has a pseudodeterministic algorithm}. Based on this intuition, we show that pseudodeterministic algorithms for \APEP\  can shed light on a few central topics in complexity theory such as circuit lowerbounds, probabilistic hierarchy theorems, and multi-pseudodeterminism.

\end{abstract}

\section{Introduction}

\paragraph{Promise Problems:}
A promise problem $\Pi$ is a pair of disjoint sets  $(\Pi_y, \Pi_n)$ of instances.  
Introduced by Even, Selman and Yacobi~\cite{EvenSelmanYacobi84}, promise problems arise naturally in several settings such as hardness of approximations, public-key cryptography, derandomization, and completeness. While much of complexity theory is based on language recognition problems (where every problem instance is either in $\Pi_y$ or in $\Pi_n$), the study of promise problems turned out be an indispensable tool that led to new insights in the area. Many interesting open questions regarding  probabilistic complexity classes can be answered when we consider their promise versions. For example significant questions such as whether derandomization of $\BPP$ implies derandomization of $\MA$, whether derandomization of $\BPP$ implies Boolean circuit lower bounds, whether derandomization of the one-sided-error class $\RP$ implies derandomization of $\BPP$, or whether probabilistic complexity classes have complete problems remain open in the traditional classes. All these questions have 
an affirmative answer if we consider their promise analogues. For example, it is known that derandomizing $\PrBPP$ implies a derandomization of $\MA$~\cite{GoldreichZuckerman11}, and also implies Boolean circuit lower bounds~\cite{ImpagliazzoKabenetsWigderson02}. Similarly, there exist promise problems that are complete for classes such as $\PrBPP$, ${\rm PromiseRP}$, and $\SZK$~\cite{SahaiVadhan03}. We refer the reader to the comprehensive survey article by  Goldreich~\cite{Goldreich06a} for a treatment on the wide-ranging applicability of promise problems.

The role of promise problems in circumventing certain deficiencies of language recognition problems is intriguing. A way to understand the gap between promise problems and languages is by considering {\em solutions} to promise problems.
A set $S$ is a solution to a promise problem $\Pi = (\Pi_y, \Pi_n)$ if $\Pi_y \subseteq S$ and $S \cap \Pi_n = \emptyset$.
A natural question is to investigate the complexity of solutions to a promise problem. Informally, we say that for a complexity class ${\cal C}$ (for example $\BPP$), ${\rm Promise} {\cal C}=\cal{C}$, if every promise problem in ${\rm Promise}{\cal C}$ has a solution in ${\cal C}$.  Intuitively, when ${\rm Promise}{\cal C}$ equals ${\cal C}$, then there is no gap between the class ${\cal C}$ and its promise counterpart.

In this paper we establish a close connection between promise problems and the seemingly unrelated notion of {\em pseudodeterminism}. More concretely, we establish that $\PrBPP = \BPP$ if and only if all probabilistic approximation algorithms can be made pseudodeterministic.  

\paragraph{Pseudodeterminism.}
The notion of a {\em pseudodeterministic} algorithm was introduced by Gat and Goldwasser~\cite{GatGoldwasser11}\footnote{Originally termed Bellagio algorithms}. Informally, a probabilistic algorithm $M$ is pseudodeterministic if for every $x$, there exists a {\em  canonical value} $v$ such that $\Pr[M(x)= v]$ is high. Pseudodeterministic algorithms are appealing in several contexts, such as distributed computing and cryptography, where it is desirable that different invocations of a probabilistic algorithm by different parties should produce the same output. In complexity theory, the notion of pseudodeterminism clarifies the relationship between search and decision problems in the context of randomized computations. It is not known whether derandomizing $\BPP$ to $\P$ implies derandomization of probabilistic search algorithms. However, $\BPP=\P$ implies derandomization of  {\em pseudodeterministic} search algorithms ~\cite{GoldreichGoldwasserRon13}.
Since its introduction, the notion of pseudodeterminism has received considerable attention. Section~\ref{prior} details prior and related work on pseudodeterminism.   

\subsection*{Our Results}

The main conceptual contribution of this paper is that the gap between  $\PrBPP$ and  $\BPP$ can be completely explained by the existence of pseudodeterministic algorithms for \APEP: the problem of approximating the acceptance probability of Boolean circuits additively.  While it is easy to design a probabilistic approximation algorithm for this problem,  we do not know whether there exists a  pseudodeterministic algorithm for this problem.
Very recently the  authors proved this problem {\em complete} for problems that admit approximation algorithms (more generally multi-pseudodeterministic algorithms as defined by Goldreich~\cite{Goldreich19}) in the context of pseudodeterminism~\cite{DPV21}. In particular, they showed that if \APEP\ admits a pseudodeterministic algorithm, then every probabilistic approximation algorithm can be made pseudodeterministic. 
Our connection between pseudodeterminism and promise problems is established via \APEP\ and is stated below.





\medskip
\noindent{\bf Result 1.}{\em \PrBPP\ has a solution in \BPP\ if and only if \APEP\ has a pseudodeterministic approximation algorithm.}

Based on the above result, we obtain results that connect pseudodeterminism to circuit lower bounds, probabilistic hierarchy theorems, and multi-pseudodeterminism. 
\paragraph{\em Circuit lower bounds:} Establishing  lower bounds against fixed polynomial-size circuits has a long history in complexity theory. In this line of work, the focus is on establishing upper bounds on the complexity of languages that can not be solved by any Boolean circuit of a fixed polynomial size. One of the central open questions in this area is to show that \NP\ has languages that cannot be solved by linear-size Boolean circuits. Over the years researchers have made steady progress on this question. Kannan~\cite{Kan82} showed that there are problems in $\Sigma_2^{\P}$ that do not have linear-size circuits (more generally, size $O(n^k)$ for any constant $k$). Later, using techniques from learning theory, this upper bound was improved to $\ZPP^{\NP}$~\cite{BCGKT96,KoblerWatanabe98} and later to  $S_2^{\P}$~\cite{Cai01}. Vinodchandran showed that the class $\PP$ does not have fixed polynomial-size circuits~\cite{Vinodchandran05}. 
Santhanam~\cite{Santhanam09} showed that further progress can be made if we relax the complexity classes to also include {\em promise classes}. In particular, he showed that
$\PromiseMA$ does not have fixed polynomial-size circuits. It is not known whether this result can be improved to  the traditional class $\MA$. 
We show that if \APEP\  has pseudodeterministic algorithms then $\MA$ has languages that can not be solved by $O(n^k)$ size circuits for any $k$.

\medskip
\noindent{\bf Result 2.} {\em If \APEP\ admits pseudodeterministic approximation algorithms, then for any $k$, there are languages in $\MA$ that do not have $O(n^k)$ size Boolean circuits}.  
\smallskip

In fact we show that under the assumption, $\MA = \exists. \BPP$ and thus  $\exists.\BPP$ does not have fixed polynomial-size circuits.  The above result improves the connection between pseudodeterministic algorithms and circuit lower bounds established in~\cite{DixonPavanVinod18}, where it was shown that designing a $\BPP^{\NP}_{tt}$ pseudodeterministic algorithm for problems in $\sharpNP$ would yield super-linear circuit lower bounds for languages in $\ZPP^{\NP}_{tt}$. 

\paragraph{\em Hierarchy theorem for probabilistic classes:}
Some of the most fundamental results in complexity theory are hierarchy theorems -- given more resources, more languages can be recognized. The time hierarchy theorem states that if $T_1(n)\log T_1(n) \in o(T_2(n))$, then there exist languages that can be decided in deterministic time $O(T_2(n))$, but not in deterministic time $O(T_1(n))$~\cite{HennieS66,SHL65}. Similar hierarchy results hold for deterministic space and nondeterministic time~\cite{Cook73,SFM78,Zak83}. Proving hierarchy theorems for probabilistic time is a lot more challenging. There has been significant work in this direction~\cite{Barak02,FortnowSanthanam04,FortnowSanthanamTrevisan05,MelkebeekP06}. All these results use an ``advice bit'', i.e. the results established are of the form ``there is a language in $\BPTIME(T_2(n))/1$ that is not in $\BPTIME(T_1(n))/1$. Removing the advice bit has been a vexing open problem. We show that a pseudodeterministic algorithm for \APEP\ leads to hierarchy theorems for bounded-error probabilistic time.

\medskip
\noindent{\bf Result 3.} {\em If \APEP\ admits pseudodeterministic approximation algorithms, then hierarchy theorems for\BPTIME\ hold. In particular ${\BPTIME(n^\alpha) \subsetneq \BPTIME(n^{\beta})}$ for constant $1\leq \alpha < \beta$.}  

\paragraph{\em Multi-pseudodeterminism:}
Goldreich observed that the problem of estimating the average value of a function over a large universe admits a {\em 2-pseudodeterministic algorithm}: a probabilistic polynomial-time algorithm that outputs {\em two canonical} values with high probability~\cite{Goldreich19}. Motivated by this, Goldreich introduced the notion of {\em multi-pseudodeterminism}~\cite{Goldreich19}. A {\em k-pseudodeterministic} algorithm is a probabilistic-polynomial time algorithm that, for every input $x$, outputs a value from a set $S_x$ of size at most $k$ with high probability (the exact probability bound has to be carefully defined, see Section~\ref{prelims} for a formal definition and \cite{Goldreich19} for justification for the definition).   

In~\cite{DPV21}, the authors show that \APEP\ is a complete problem for functions that admit $k$-pseudodeterministic algorithms for any constant $k$, in the sense that such functions admit pseudodeterministic algorithms if \APEP\ admits a pseudodeterministic algorithm. Here we improve this result to functions that admit $k$-pseudodeterministic algorithms for any polynomial $k$. 

\medskip
\noindent{\bf Result 4.}
{\em If \APEP\ admits a pseudodeterministic approximation algorithm, then
 every multi-valued function $f$ that admits a $k(n)$-pseudodeterministic algorithm, for a polynomial $k(n)$, is in {\rm Search BPP}. Moreover under the assumption, every multi-valued function $f$ that admits a $k(n)$-pseudodeterministic algorithm also admits a pseudodeterministic algorithm}. 

\paragraph{\em Concurrent Work:} In an independent and recent work, Lu, Oliveria, Santhanam~\cite{LuOlivieraSanthanam21} also explored the consequences of pseudodeterministic algorithms for \APEP\ (they use CAPP to denote \APEP). There is some intersection between their work and ours. In particular, they also establish results on probabilistic hierarchy. 
They showed that if  there is a pseudodeterministic algorithm  for \APEP\ that is correct on average at infinitely many input lengths, then the hierarchy theorems for $\BPTIME$ follow. Note that our work  considers existence of  pseudodeterminitic algorithms for \APEP\  in the worst-case. The rest of the work is different. 
Their work has results that include designing pseudodeterministic  pseudorandom  generators, and an equivalence between probabilistic hierarchy theorems and pseudodeterministic algorithms for constructing strings with large $rKt$ complexity, which we do not have.
Their work did not explore the relationships of  pseudodeterministic algorithms with promise problems, circuit lowerbounds, and multi-pseudodeterminism which we establish.

\ignore{
Consider the search problem of producing a witness that  two multi-variate polynomials $f$ and $g$ over a field are different. A simple probabilistic polynomial-time algorithm for this problem randomly picks an element  $t$ from the domain and outputs it if $f(t)\neq g(t)$.  Even though this algorithm is simple and efficient and the error probability can be made arbitrarily small, Gat and Goldwasser~\cite{GatGoldwasser11} pointed out a deficiency: two different runs of the algorithm can produce two different witnesses with very high probability. Well-known probabilistic algorithms for several search problems, such as finding a large prime number or computing generators of cyclic groups, also exhibit this deficiency. This raises the question of whether we can design a  probabilistic algorithm for search problems that will output the same witness on multiple executions, with high probability.}

\subsection{Prior and Related Work on Pseudodeterminism}\label{prior}
One line of research on pseudodeterminism  has focused on designing pseudodeterministic algorithms for concrete problems. Gat and Goldwasser designed polynomial-time pseudodeterministic algorithms for various  algebraic problems such as finding quadratic non-residues and finding non-roots of multivariate polynomials~\cite{GatGoldwasser11}. Goldwasser and Grossman exhibited a pseudodeterministic NC algorithm for computing matchings in bipartite graphs~\cite{GoldwasserGrossman17}. Recently, Anari and Vazirani~\cite{AV20} improved this result general graphs.  Grossman designed a pseudodeterministic algorithm for computing primitive roots whose runtime matches the best known Las Vegas algorithm~\cite{Grossman15}. Oliveira and Santhanam~\cite{OliveiraSanthanam17} designed  a sub-exponential time pseudodeterministic algorithm for generating primes that works at infinitely many input lengths. Subsequently, Oliveira and Santhanam also showed that \APEP\ admits a subexponential-time pseudodeterministic algorithm that is correct on average at infinitely many input lengths~\cite{OliveiraSanthanam18}.
Goldreich, Goldwasser and Ron~\cite{GoldreichGoldwasserRon13}, and later Holden~\cite{Holden17}, investigated the possibility of obtaining pseudodeterministic algorithms for BPP search problems. 

Other lines of  work  extended the notion of  pseudodeterminism to several other scenarios including interactive proofs,  streaming and sublinear algorithms, and learning algorithms ~\cite{GoldwasserGrossmanHolden17,GoemansGoldwasserHolden20, GGMW20, GoldreichGoldwasserRon13, OliveiraSanthanam18}.  
The works of Grossman and Liu, and Goldreich introduced generalizations of pseudodeterminism such as {\em reproducible algorithms}, {\em influential bit algorithms}, and {\em multi-pseudodeterministic algorithms}~\cite{GrossmanLiu19,Goldreich19}.
Very recently the  authors exhibited {\em complete problems} for functions that admit approximation algorithms, more generally multi-pseudodeterministic algorithms as defined by Goldreich~\cite{Goldreich19}, in the context of pseudodeterminism~\cite{DPV21}.


\section{Preliminaries}
\label{prelims}
In this paper, we are concerned with additive error approximations. A probabilistic algorithm $A$ is an $(\varepsilon, \delta)$-additive approximation algorithm for a function $f:\{0,1\}^* \rightarrow \mathbb{R}$ if the probability that $A(x) \in [f(x) - \varepsilon$, $f(x) + \varepsilon]$ is at least $1-\delta$.

\subsection{Pseudodeterminism}
\begin{definition}
\noindent{\sc Acceptance Probability Estimation Problem:} $\APEP_{(\varepsilon,\delta)}:$ Given a Boolean circuit $C:\{0, 1\}^n \rightarrow \{0, 1\}$, give an $(\varepsilon,\delta)$-additive approximation 
for $\Pr_{x \in U_n} [C(x) = 1]$.
\end{definition}

\begin{definition}[\cite{GatGoldwasser11},\cite{Goldreich19}]

Let $f$ be a multivalued function, i.e. $f(x)$ is a non-empty set. We say that $f$ admits pseudodeterministic algorithms if there is a probabilistic polynomial-time algorithm $A$ such that for every $x$, there exists a $v\in f(x)$ such that $A(x) = v$ with probability at least $2/3$.
$f$ admits $k$-pseudodeterministic algorithms if there is a probabilistic polynomial-time algorithm $A$ such that for every $x$, there exists a set $S_x \subseteq f(x)$ of size at most $k$ and the probability that $A(x) \in S(x)$ is at least $\frac{k+1}{k+2}$.
\end{definition}
Note that the above definition captures pseudodeterminism for approximation algorithms, as approximation algorithms can be viewed as multivalued functions.
It is known that any function that admits an $(\varepsilon,\delta)$ approximation algorithm admits a $(2\varepsilon,\delta)$ 2-pseudodeterministic algorithm (see~\cite{Goldreich19,DPV21} for a proof). 

\begin{proposition}\label{2pseudo}
For every $0< \varepsilon,\delta< 1$, there is a 2-pseudodeterministic algorithm for $\APEP_{(\varepsilon,\delta)}$.
\end{proposition}

Gat and Goldwasser proved the following  characterization~\cite{GatGoldwasser11}.

\begin{theorem}\label{thm:GGR}
A function admits a pseudodeterministic algorithm if and only if it is computable in $\PF^{\BPP}$.  
\end{theorem}

\begin{definition}[\SearchBPP~\cite{Goldreich11}]
 A search problem is a relation $R \subseteq \{0,1\}^* \times \{0,1\}^*$. For every $x$, the witness set $W_x$ of $x$ with respect to $R$ is $W_x=\{y\mid(x,y) \in R\}$. 
 A search problem  $R$ is in $\SearchBPP$ if
 \begin{enumerate}
    \item For every $x$, there is an efficient probabilistic algorithm to output an element of $W_x$: i.e. 
    there exists a probabilistic polynomial-time algorithm
$A$ such that for every $x$ for which $W_x\ne \phi$, $A(x) \in W_x$  with probability $\geq 2/3$, and 
\item $R\in \BPP$: i.e.
there exists a probabilistic polynomial-time algorithm $B$ such that if $(x,y)\in R$, then $B(x,y)$ accepts with probability $>2/3$, and if $(x,y)\not\in R$ then $B(x,y)$ accepts with probability $<1/3$.
\end{enumerate}
\end{definition}

\begin{definition}
For a multivalued function $f$, we say that $f$ is in $\SearchBPP$ if there is a relation $R$ in  $\SearchBPP$ so that $\forall x$, the witness set $W_x \neq \phi$ and  $W_x \subseteq f(x)$. 
\end{definition}

Dixon, Pavan and Vinodchandran~\cite{DPV21} proved that \APEP\  is a complete problem for pseudodeterministic approximation algorithms and pseudodeterministic $\SearchBPP$ in the following sense.

\begin{theorem}\label{thm:complete}
If $\APEP_{(1/100,1/8)}$  admits a pseudodeterministic algorithm then 
\begin{enumerate}
\item every function $f$ that has an  $(\varepsilon,\delta)$-approximation algorithm has a pseudodeterministic $(3\varepsilon, \delta)$-approximation algorithm. 
\item every problem in \SearchBPP\  has a pseudodeterministic algorithm.
\end{enumerate}
\end{theorem}

It is well known that for every  $0 < \varepsilon,\delta < 1$, there is a probabilistic algorithm for $\APEP_{(\varepsilon,\delta)}$ that runs in time $\poly(n,1/\varepsilon, \log 1/\delta)$ where $n$ is the input length.  Thus by the above result, we obtain the following proposition.

\begin{proposition}\label{prop:complete}
If $\APEP_{(1/100,1/8)}$ has a pseudodeterministic algorithm then for every $0< \varepsilon,\delta < 1$, $\APEP_{(\varepsilon,\delta)}$ has a pseudodeterministic algorithm.  
\end{proposition}

\noindent{\bf\em Remark.} In the rest of the  paper, we use the phrase ``\APEP\ has a pseudodeterministic algorithm'' in place of ``$\APEP_{(1/100, 1/8)}$ admits a pseudodeterministic algorithm'', and denote the presumed pseudodeterministic algorithm with  \algoapep.

\subsection{Promise Problems}
\begin{definition}
A promise problem $\Pi= (\Pi_y, \Pi_n) \in \PrBPP$ if  there exists a probabilistic polynomial-time machine $M$ such that $\forall x$
\[x\in \Pi_y \Leftrightarrow \Pr[M(x) = \mbox{ accepts}] \geq 2/3,\]
\[x\in \Pi_n \Leftrightarrow \Pr[M(x) = \mbox{ accepts}] <1/3,\]
\end{definition}

We can similarly define promise classes such as $\PromiseMA$.

\begin{definition}
Let ${\cal C}$ be a complexity class. We say that a promise $(\Pi_y, \Pi_n)$
has a solution in ${\cal C}$ if there exists a language $L$ in ${\cal C}$ such that $\Pi_y \subseteq L$ and $L \cap \Pi_n = \emptyset$.
\end{definition}
\begin{definition}

Let $ \Pi = (\Pi_y,\Pi_n)$ be a promise problem. $\Pi' =  \exists  \cdot \Pi$ is a promise problem $(\Pi'_y,\Pi'_n)$ defined as follows. There is a polynomial $p$ such that $\forall x$

\[x \in \Pi'_y \Leftrightarrow  \exists w \in \{0,1\}^{p(|x|)}, \langle x,w \rangle \in \Pi_y 
\]
\[
 x \in \Pi'_n  \Leftrightarrow \forall w \in \{0,1\}^{p(|x|)}, \langle x,w \rangle \in \Pi_n
\]
\end{definition}
\begin{definition}
We say that a promise problem $\Pi = (\Pi_y, \Pi_n) \in \exists \cdot {\PrBPP}$ if there is a promise problem $\Pi' \in \PrBPP$ such that $\Pi = \exists \cdot \Pi'$.
\end{definition}

\begin{definition}
A probabilistic polynomial-time machine $M$ has $\BPP$-type behaviour if on every input $x$, $\Pr[M(x) \mbox{ accepts} ]$ is either $\geq 2/3$ or $< 1/3$.
\end{definition}

\section{Consequences of Pseudodeterministic Algorithm for \APEP}

\subsection{Promise Problems}
\begin{theorem}\label{thm:prbpp-bpp}
$\PrBPP$ has a solution in $\BPP$ if and only if \APEP\ has a pseudodeterministic approximation algorithm.
\end{theorem}

\begin{proof}
$(\Leftarrow):$
We will first prove that if \APEP\ has a pseudodeterministic algorithm, then $\PrBPP$ has a solution in $\BPP$.  Let $\Pi$ be a promise problem in $\PrBPP$ and let $M$ be a probabilistic polynomial-time machine that witnesses this. Given $x$, let $C_x$ be the following Boolean circuit: 
\[
C_x(r) = 1 \mbox{ if and only if $M(x)$ on random string $r$ accepts}.  
\]

Note that given $x$, we can  construct $C_{x}$ in time $\poly(|x|)$.  Consider the following probabilistic algorithm:

\medskip
\noindent{\bf Algorithm $B$:}
On input $x$, construct $C_{x}$ and run $\algoapep(C_{x})$. If $\algoapep(C_{x})\geq 1/2$, accept; else reject.
\smallskip

\begin{claim}
$B$ has a  $\BPP$-type behavior.
\end{claim}

\begin{proof}
Let $x$ be an input to $B$. Recall that $\algoapep$ is a pseudodeterministic approximation algorithm that outputs a canonical value $v$ on input $C_x$ with probability at least $7/8$. So either with probability at least $7/8$, $v$ is $\geq 1/2$, in which case $B$ accepts $x$, or with probability at least $7/8$, $v$ is $< 1/2$ and $B$ rejects. Thus for every input $x$, $B$ either accepts with probability $\geq 7/8$ or rejects with probability $\geq 7/8$, and thus $B$ has $\BPP$-type behaviour. 
\end{proof}

Let $L$ be the language accepted by the above machine. Then by the above claim $L\in \BPP$.

\begin{claim}
$L$ is a solution to the promise problem $\Pi$.
\end{claim}

\begin{proof}
Let $x$ be a string in $\Pi_y$. Thus $\Pr[C_x(r)= 1] \geq 2/3$. Thus $\algoapep(C_{x})$ outputs a canonical value $v \geq 2/3 - 1/100 > 1/2$ with probability at least $7/8$, and thus $B$ accepts with probability at least $7/8$, and thus $x \in L$.

Suppose $x \in \Pi_n$. Thus 
Thus $\Pr[C_x(r)= 1] < 1/3$.
Thus $\algoapep(C_{x})$ outputs a canonical value $v \leq  1/3 + 1/100 < 1/2$ with probability at least $7/8$, and thus $B$ rejects  with probability at least $7/8$, and thus $x \notin L$.
\end{proof}

By the above two claims we obtain that if \APEP\ has a pseudodeterministic approximation algorithm, $\PrBPP$ has a solution in $\BPP$. 

\smallskip
\noindent{($\Rightarrow$):} Now suppose that $\PrBPP$ has a solution in $\BPP$. By Proposition~\ref{2pseudo}, there is a  2-pseudodeterministic $(\varepsilon, \delta)$ approximation algorithm  $M$ for \APEP\ where $\delta = 1/4$ and $\epsilon = 1/200$. We slightly modify $M$ as follows: whenever $M$ outputs a value $v$, then output a value $v'$ that is the closest integer multiple of $\varepsilon$ to $v$. Note that the modified machine $M$ is a $(2\varepsilon, \delta)$ approximation algorithm for \APEP. The machine $M$ has the property that every output is of the form $k\varepsilon$, $0 \leq k \leq 1/\varepsilon$. 

For a Boolean circuit $C$, let $p_C$ denote the acceptance probability of  $C$. Thus for every $C$, we have
\begin{equation}\label{eqn1}
\Pr[M(C) \in (p_C - 2\varepsilon, p_C + 2\varepsilon)] \geq 3/4 
\end{equation}

We associate a promise problem $\Pi= (\Pi_y, \Pi_n)$ with $M$. This definition of promise problem is inspired by the work of Goldreich~\cite{Goldreich11}.  

\[\Pi_y = \{\langle C, v\rangle ~|~\mbox{$M(C)$ outputs $v$ with probability at least $3/8$}\}\]
\[\Pi_n =\{\langle C, v\rangle ~|~\mbox{$M(C)$ outputs $v$ with probability at most $1/4$}\}\]

We make the following two critical observations. 
\begin{observation}\label{obs:trivial}
If $\langle C, v\rangle \notin \Pi_n$, then $v \in (p_C - 2\varepsilon, p_C+2\varepsilon)$.
\end{observation}

This observation follows from equation~\ref{eqn1}. 

\begin{observation}\label{obs:trivial2}
For every Boolean circuit $C$, there exists a $v$ such that $\langle C, v\rangle \in \Pi_y$ and $v = k \epsilon$ for some $k > 0$, 
\end{observation}

\begin{proof}
Since $M$ is 2-pseudodeterministic, there is a set $S$ of size at most $2$ such that every element in $S$ lies between $p_C-2\epsilon$ and $p_C+2\epsilon$ and $\Pr[M(C) \in S] \geq 3/4$. Thus there must exist an element $v$ from $S$ such that $M(C)$ outputs $v$ with probability at least $3/8$. Finally note that the modification of $M$ described earlier ensures that $M$ always outputs a multiple of $\epsilon$.
\end{proof}

\begin{claim}
$\Pi \in \PrBPP.$
\end{claim}
\begin{proof}
Consider the algorithm $M_{\Pi}$: On input $\langle C, v\rangle$ run $M(C)$. If it outputs $v$, then accept, else reject. This algorithm accepts all instances from $\Pi_y$ with probability at least $3/8$ and accepts  all instances from $\Pi_n$ with probability at most $1/4$. Since there is a gap between $3/8$ and $1/4$, this gap can be amplified with standard amplification techniques. This implies that $\Pi$ is in $\PrBPP$.
\end{proof}

Now we will complete the proof by designing a pseudodeterministic algorithm for \APEP. By our assumption there is a language $L_{\Pi} \in \BPP$\  that is a solution to $\Pi$.  Consider the following deterministic algorithm for \APEP\ with oracle access to $L_{\Pi}$. On input $C$, check if $\langle C, k\varepsilon\rangle \in L_{\Pi}$ for integer values of $k$, $0 \leq k \leq 1/\varepsilon$. 
Let $\ell$ be the first value such that $\langle C, \ell \varepsilon\rangle \in L_{\Pi}$, then  output $\ell \varepsilon$. By Observation~\ref{obs:trivial2}, such an $\ell$ must exist.  
Moreover, if $\langle C, \ell \varepsilon \rangle \in L_{\Pi}$, then it must be the case that $\langle C, \ell \varepsilon \rangle \notin \Pi_n$. By Observation~\ref{obs:trivial}, we have that $\ell \varepsilon \in (p_c + 2\varepsilon, p_c - 2\varepsilon)$. Thus \APEP\ has a $(2\varepsilon, \delta)$, $\PF^{\BPP}$ approximation algorithm.  This  implies that \APEP\ has a $(2\varepsilon, \delta)$ pseudodeterministic algorithm by Theorem~\ref{thm:GGR}.

\end{proof}

We obtain the following corollary by using the completeness result of \APEP\ . 

\begin{corollary}
If $\PrBPP$ has a solution in $\BPP$, then $\SearchBPP$ admits pseudodeterministic algorithms.
\end{corollary}
\begin{proof}
From the above theorem, if $\PrBPP$ has a solution in $\BPP$, then \APEP\ has pseudodeterministic algorithms.
The proof follows from Theorem~\ref{thm:complete}.
\end{proof}

\subsection{Circuit Lower Bounds}
\begin{theorem}
If \APEP\ admits pseudodeterministic approximation algorithms, then 
\begin{enumerate}
 \item Every promise problems $\Pi = (\Pi_Y, \Pi_N)$ in \PromiseMA\ has a solution in \MA. \label{two}
     \item $\MA = \exists \cdot \BPP$. \label{three}
   \item $\MA$ does not have fixed polynomial-size circuits.
   \end{enumerate}
\end{theorem}
\begin{proof}
\begin{enumerate}

\item We first show that if $\Pi$ is a promise problem in $\PromiseMA$, then $\Pi \in \exists \cdot \PrBPP$. Let $M$ be a probabilistic polynomial-time  verifier.  Consider the following promise problem $\Pi'$: A tuple $\langle x, w\rangle$ is a positive instance if  $M$ accepts $\langle x, w\rangle$ with probability at least $2/3$ and is a negative instance if $M$ accepts $\langle x, w\rangle$ with probability at most $1/3$. It is easy to see that $\Pi  = \exists. \Pi'$.  By Theorem~\ref{thm:prbpp-bpp}, $\Pi'$ has a solution $L'$ in $\BPP$ if \APEP\ admits pseudodeterministic algorithms.  Note that the language $L = \exists \cdot L'$ is a solution to $\Pi$, and $\exists \cdot L'$ is in $\exists \cdot \BPP$.  Since $\exists \cdot \BPP$ is a subset of $\MA$, the claim follows. 

\item The above proof showed that every promise problem in $\PromiseMA$ has a solution in $\exists\cdot \BPP$. Thus it follows that $\MA = \exists \cdot \BPP$.

\item Santhanam~\cite{Santhanam09} showed that for every $k$, there is a problem $\Pi_k$ in $\PromiseMA$ that does not have any solution that admits $O(n^k)$ size circuits. Since by Item~\ref{two} $\Pi_k$ has a solution $L_k \in \MA$, we get that $L_k$ does not have $O(n^k)$ size circuits.  Combining this with~\ref{three}, it follows that $\exists \cdot \BPP$ does not have $O(n^k)$ size circuits.
\end{enumerate}
\end{proof}

The above result reveals an interesting connection between pseudodeterminism, derandomization of $\BPP$, and circuit complexity. If \APEP\ has pseudodeterministic algorithms, then derandomizing $\BPP$ to $\P$ implies that $\NP$ does not have fixed polynomial-size circuits.

\subsection{Hierarchy Theorems}
\begin{theorem}
If \APEP\ admits pseudodeterministic approximation algorithms, then hierarchy theorems for \BPTIME\ hold. In particular, ${\BPTIME(n^\alpha) \subsetneq \BPTIME(n^{\beta})}$ for constant $1\leq \alpha < \beta$.  
\end{theorem}
\begin{proof}

We will first show that there is a constant $c$ so that $\BPTIME(n) \subsetneq \BPTIME(n^c)$.  A similar arguments will show that $\BPTIME(n^a) \subsetneq \BPTIME(n^{ca})$ for every $a >0$. Then the theorem will follow from padding arguments. 
 
Let $\{M_i\}_{i \geq 1}$ be an enumeration of probabilistic linear-time Turing machines. Suppose that $\algoapep$ runs in time $m^a$ in circuits of size $m$ for some $a > 0$.
For every $M_i$, consider a probabilistic machine $M'_i$ defined as follows. 
$M'_i$ on input $x$ constructs a circuit $C_{i,x}$ as follows. 
The circuit $C_{i,x}$ on input $r$ simulates $M_i(x)$ with $r$ as random bits and accepts if and only if $M_i$ accepts. Now $M'_i$ runs $\algoapep(C_{i,x})$,
and accepts if  and only if the output of $\algoapep(C_{i,x}) \geq 1/2$.

\begin{claim}
There exists a constant $c > 0$ such that for every $i$, the machine $M'_i$ runs in time $O(n^c)$. 
\end{claim}

\begin{proof}
Let $n$ be the length of input $x$ to $M'_i$. The machine $M'_i(x)$ first constructs the circuit $C_{i,x}$. Since $M_i(x)$ runs in $O(n)$ time, the size of the circuit $C_{i,x}$ is bounded by $O(n^2)$, it can be constructed in $O(n^2)$ time. Next $M'_i$ runs $\algoapep$ on $C_{i,x}$, this steps takes $O(n^{2a})$ time. Since $a$ is a constant, there is a universal constant $c$ such that the runtime of $M'_i$ is $O(n^c)$.

\end{proof}

\begin{claim}
For every $i$, $M'_i$ has $\BPP$-type behaviour
\end{claim}

\begin{proof}
This follows because the pseudodeterministic algorithm $\algoapep$, on every input, outputs a canonical value $v$ with probability at least $2/3$. If the canonical value $v \geq 1/2$, then $M_i'$ accepts with probability at least $2/3$, else $M_i'$ accepts with probability at most $1/3$.  Thus $M_i$ has $\BPP$-type behaviour.
\end{proof}

\begin{claim}\label{index-claim}
For every $L \in \BPTIME(n)$, there is $i > 0$ such that $M'_i$ accepts $L$.
\end{claim}

\begin{proof}
Since $L \in \BPTIME(n)$, there exists an $i > 0$ such that $M_i$ accepts $L$ and $M_i$ has $\BPP$-type behaviour. Let $x \in L $ be an input to $M_i$. The probability that $M_i$ accepts is $\geq 2/3$. Thus the acceptance probability of the circuit $C_{i,x}$ is at least $2/3$. Thus $\algoapep(C_{i,x})$ outputs a canonical $v \geq 2/3-1/100 \geq 1/2$ with probability at least $2/3$. Thus $M'_i$ accepts $x$ with probability $\geq 2/3$.
Similar arguments show that if $x \not\in L$, $M'_i$ rejects $x$ with probability $\geq 2/3$.

\end{proof}

Now using the standard diagonalization argument, we construct a language  $L_D$ in $\BPTIME(n^{c+1})$. The language $L_D$ is a tally language and we describe it via a $\BPTIME(n^{c+1})$ machine $N$ that accepts it. The machine $N$ on input $0^i$ simulates $M'_i(0^i)$ and accepts if and only if $M'_i(0^i)$ rejects. 
Since $M'_i$ has $\BPP$-type behaviour, $N$ also has $\BPP$-type behaviour. Since $M'_i$ runs in time $O(n^c)$, $N$ can simulate it in time $O(n^{c+1})$. Thus $L_D \in \BPTIME (n^{c+1})$. Suppose that $L_D\in \BPTIME (n)$. By Claim~\ref{index-claim}, there exists $i > 0$ such that $M'_i$ accepts $L_D$. Now consider input $0^i$. Observe that $M'_i$ accepts $0^i$ if and only if $N$ rejects $0^i$. Thus $0^i \in L_D$ if and only if $M'_i$ rejects $0^i$. This is a contradiction, and thus $L_D \notin \BPTIME(n)$. 
\end{proof}

\subsection{Multivalued Functions}

\begin{theorem}
If \APEP\ admits pseudodeterministic approximation algorithms, then
 every multivalued function $f$ that admits a $k(n)$-pseudodeterministic algorithm for a polynomial $k(n)$ is in {\SearchBPP}.
\end{theorem}
\begin{proof} 
 Let $f$ be a multi-valued function and let $M_f$
 be a $k(n)$-pseudodeterministic algorithm for $f$.
Without loss of generality we can assume that $f$ maps strings of length $n$ to strings of length $p(n)$ for some polynomial $p$.
 For input $x$ of length $n$, let $S_x$ be the set of size $\leq k(n)$ such that $S_x \subseteq f(x)$ and $M_f(x) \in S_x$ with probability $\geq \frac{k(n)+1}{k(n)+2}$. From the definition of $k(n)$-pseudodeterminism, we have the following claim. 

\begin{claim}\label{clm:vstar}
$\exists v* \in S_x$ such that $\Pr \left [M_f(x) = v* \right ] \geq \frac{1+1/k(n)}{k(n)+2}$. Moreover,  $\forall v \not\in S_x$  $\Pr[M_f(x)=v] < \frac{1}{k(n)+2}$. 
\end{claim}

Let $\tau = \frac{1+ 1/2k(n)}{k(n)+2}$ be a threshold that is the middle point of $\frac{1+1/k(n)}{k(n)+2}$ and $\frac{1}{k(n)+2}$. 
For a pair of strings $\langle x,v\rangle$, where $|x| = n$ and $|v| = p(n)$,  let $C_{x,v}$ be the following Boolean circuit. $C_{x,v}$ on input $r$, outputs 1 if $M_f(x)$ on random string $r$ outputs $v$, 0 otherwise.  We will show that  there is a relation $R$ so that  (1) $\forall x: W_x \neq \phi $ and $W_x \subseteq f(x)$, and (2) $R\in  \SearchBPP$. We define the relation $R$ as follows. 

\[
R= \{\langle x, v \rangle \mid \mbox{the canonical output of } \algoapep(C_{x,v}) \geq \tau \} 
\]

Here $\algoapep$ is the $(\varepsilon, \delta)$ pseudodeterministic algorithm for \APEP, where $\varepsilon = {1/2k(n)(k(n)+2)}$ and $\delta = 2^{-n}$. Note that such an algorithm exists under the assumption by Proposition~\ref{prop:complete} and standard error reduction techniques. 

\begin{claim}
$\forall x, W_x \subseteq f(x)$ and $W_x$ is not empty. 
\end{claim}

\begin{proof}
For this we show that $W_x \subseteq S_x$. If $v \not\in S_x$, then $M_f(x)$ outputs $v$ with probability at most $1/(k(n)+2)$, thus the canonical output of $A(C_{x,v})$ is $<\frac{1}{k(n)+2}+ \varepsilon = \tau$ and by definition $v\not\in W_x$. On the other hand, Since $v^* \in W_x$, the canonical output of $\algoapep(C_{x, v^*})$ is $\geq \frac{1+1/k(n)}{k(n)+2} - \varepsilon = \tau$. Thus $v^* \in W_x$. Thus $W_x \neq \phi$
\end{proof}

\begin{claim}
$R \in \BPP$. 
\end{claim}
\begin{proof}
Consider the algorithm that on input $\langle x,v \rangle$, runs $\algoapep(C_{x,v})$ and accepts if and only if the output of $\algoapep$ is $\geq \tau$. 
Since $\algoapep$ is a pseudodeterministic algorithm for \APEP\, it outputs a canonical value with probability at least $1- 1/2^n$. This shows that $R$ is in $\BPP$.
\end{proof}

\begin{claim}
There is a probabilistic algorithm $B$
 that on input $x$ outputs $v\in W_x$ with probability $> 2/3$. 
\end{claim}

\begin{proof}
We first design an algorithm $B'$ with a nontrivial success probability and boost it to get algorithm $B$.

\noindent{\bf Algorithm $B'$:} On input $x$, run $M_f(x)$. Let $v$ be an output. Construct circuit $C_{x,v}$ and run $\algoapep$ on $C_{x,v}$. If the output of $\algoapep$ is $\geq \tau$, output $v$. Otherwise output $\bot$.   
\smallskip

Consider a  $v \notin W_x$. Then by definition of $R$, we have that the canonical output of $\algoapep(C_{x, v})$ is less than $\tau$. Thus $\algoapep(C_{x, v})$ outputs a value larger than $\tau$ with probability at most $1/2^n$.   Thus we have that for every $v \notin W_x$
\[
\Pr[B' \mbox { outputs $v$} | M_f(x) \mbox{ outputs $v$}] \leq 1/2^n
\]

\begin{eqnarray*}
\Pr[B' \mbox{ outputs a  $v \notin W_x$}] & =
& \sum_{v \notin W_x} \Pr[B' \mbox{ outputs $v$} | M_f(x) \mbox{ outputs  $v$}] \times \Pr[M_f(x) \mbox{outputs $v$} ] \\
& \leq & 1/2^n \sum_{v} \Pr[M_f(x) \mbox{outputs $v$} ]\\
& \leq & 1/2^n
\end{eqnarray*}


By Claim~\ref{clm:vstar},  probability that $M_f(x)$ outputs $v^*$ is at least $\frac{1+1/k(n)}{k(n)+2}$, it must be the case that the canonical output of $A(C_{x, v^*})$ is at least $\frac{1+1/k(n)}{k(n) + 2} - \varepsilon = \tau$.  Thus $v^* \in W_x$.  Thus $\algoapep(C_{x,v^*})$ outputs a value $\geq \tau$ with probability at least $1 - 1/2^n$.
 Thus the probability that $B'$ outputs $v^*$ is at least $\frac{1+1/k(n)}{k(n)+2} \times (1 - 1/2^n)$.
 
 Thus $B'$ outputs a value that is not in  $W_x$ with probability at most $1/2^n$, it outputs a value in $W_x$ with probability at least $\frac{1+1/k(n)}{k(n)+2} \times (1 - 1/2^n)$, and  outputs $\bot$ with the remaining probability. We obtain $B$ by repeated invocations  ($O(k(n)^3)$ many) of $B'$ and outputting the most frequent output. 

 \end{proof}

This completes the proof that $f$ is in $\SearchBPP$.

\end{proof}
Using the above result, we obtain the following corollary, which improves a result from~\cite{DPV21}.

\begin{theorem}
If \APEP\ admits pseudodeterministic algorithm, then any multivalued function  that admits a $k(n)$-pseudodeterministic algorithm also admits a
pseudodeterministic algorithms, where $k(n)$ is a polynomial.
\end{theorem}
\begin{proof}
By the above theorem, if \APEP\ admits pseudodeterministic algorithm, then any problem that admits a $k(n)$-pseudodeterministic algorithm is in $\SearchBPP$. By Theorem~\ref{thm:complete},  if \APEP\  admits pseudodeterministic algorithms, every problem 
in $\SearchBPP$ has a pseudodeterministic algorithm.
\end{proof}

~~\\
\noindent{\bf Acknowledgements.}  We thank Zhenjian Lu, Igor Oliveira, and Rahul Santhanam for sharing a draft of their work.

\bibliographystyle{alpha}

\newcommand{\etalchar}[1]{$^{#1}$}

\end{document}